\documentclass[12pt]{amsart}
\usepackage{latexsym,amsmath,amsfonts,amscd,amssymb}
\usepackage{graphicx}
\usepackage{xcolor}
\newtheorem{theorem}{Theorem}[section]
\newtheorem{theorem*}{Theorem}

\newtheorem{proposition}[theorem]{Proposition}
\newtheorem{proposition*}[theorem*]{Proposition}

\newtheorem{corollary*}[theorem*]{Corollary}

\theoremstyle{remark}

\newtheorem{remark*}[theorem*]{Remark}

\newtheorem{note*}[theorem*]{Note}

\def\bitcoinA{%
  \leavevmode
  \vtop{\offinterlineskip 
    \setbox0=\hbox{B}%
    \setbox2=\hbox to\wd0{\hfil\hskip-.03em
    \vrule height .3ex width .15ex\hskip .08em
    \vrule height .3ex width .15ex\hfil}
    \vbox{\copy2\box0}\box2}}

\begin{document}

\title{Proof of Reputation}


\author{Cyril Grunspan}
\address{L{\'e}onard de Vinci, P{\^o}le Univ., Research Center, Paris-La D{\'e}fense, France}
\email{cyril.grunspan@devinci.fr}
\author{Ricardo P\'erez-Marco}
\address{CNRS, IMJ-PRG,  Paris, France}
\email{ricardo.perez.marco@gmail.com}
\address{\tiny Author's Bitcoin Beer Address (ABBA)\footnote{\tiny If you find this article useful, you can send some anonymous 
satoshis to support our research at the pub.}:
1KrqVxqQFyUY9WuWcR5EHGVvhCS841LPLn} 
\address{\includegraphics[scale=0.3]{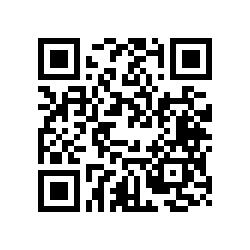}}

\begin{abstract} 
We present the new mining protocol Proof-of-Reputation (PoR) for decentralized Proof-of-Work (PoW) blockchains, 
in particular for Bitcoin. 
PoR combines the classical PoW with the new ingredient of cryptographic reputation. 
The same level of security compared to pure PoW can be achieved with a significant energy consumption 
reduction (of the order of 30\%) for the same security level.  
The proper implementation of a decentralized reputation protocol is suitable with an extra layer of mining 
security: Certified Mining. 
\end{abstract}

{\maketitle}

\section{Introduction and background}

\subsection{Proof-of-Work and decentralization.}

The aim of the present article is to enhance the classical Proof-of-Work (PoW) at the base of
the decentralized consensus algorithm which runs Bitcoin protocol invented by  
Satoshi Nakamoto \cite{N08}. 
Although we have in mind the Bitcoin protocol, our analysis is not restricted to it 
and applies to blockchains based on PoW. 


One of the major criticism of Bitcoin is its important energy use. The security of the protocol 
relies on this energy consumption. Some fundamental reasons, 
based on the Second Law of Thermodynamics, indicate that there is no
viable satisfactory shortcut without energy consumption in order to achieve the same level of security 
and full decentralization. Other propositions as 
Proof-of-Stake (PoS) lack of the main security features as Bitcoin, for example an attacker can rebuild 
the full blockchain, therefore monitorization by a central entity is constantly necessary. So far, 
no blockchain with PoS does function without supervision as Bitcoin does. 


A criterion for a blockchain to be fully decentralized is that no entity or individual is able to stop it.
Observe that this entails a level of cybersecurity unseen elsewhere, for example, no hacking 
can stop it. Also this implies that the network is permisionless and no entity is able to censor transactions. 
For a PoW based blockchain this is true as long as no one is 
in control of more than 50\% of the 
available computing power running the PoW algorithm, which, assuming homogeneous hardware efficiency is equivalent to more 
than 50\% of the energy in use.

\subsection{Decentralization and reputation.}
The energy consumption defends against Sybil attacks of the network, that is, it prevents 
the network being flooded by dishonest validators (miners). 
In an open protocol, decentralization requires that anyone is free to join the network. Censorship
is incompatible with decentralization. This is the reason 
that the so-called ``private blockchains'', where participation 
is approved by a higher entity, cannot be decentralized blockchains. The simplest solution in order to 
avoid a Sybil attack would be to weight the voting rights of any participant (node) by some system 
of reputation. Unfortunately, so far, no decentralized reputation system has been invented. Reputation 
is a subtle notion difficult to implement algorithmically. 


In Nakamoto's PoW proposal, the voting weight of each participant 
is directly proportional to the work provided. This is possible thanks to the pseudo-random properties 
of the hash function used in Bitcoin's PoW algorithm. It has the advantage that we don't need any ``memory'' 
or knowledge of past behavior of the participant. It has the drawback that no reputation of
past well behavior is incorporated. This presents also a security weakness 
since anyone without any prior credentials  
that can achieve more than 50\% of the energy dedicated to mining, and can take control of the whole network.
In some sense, Nakamoto's PoW replaces reputation by actual verifiable work being done in the present. 

As was observed in bitcointalk forum, 
it is interesting to point out that in the early 
Bitcoin code, Satoshi was trying to implement a reputation system (\cite{BitcoinCode}) based on mining 
as his comment suggests \textit{``Add atoms to user reviews for coins created''}. More recently, in the context of 
Proof of Stake, there have been propositions of reputation system as RepuCoin (\cite{repucoin}).


Thinking along these lines, we explore how to go one step further. We want to encode reputation mathematically 
with a secure cryptographic procedure. The essential nature of reputation is based on proper past behavior. 
The most common way of acquiring reputation is by past verifiable work. Other ways to acquire 
reputation is also by endorsement by reputed individuals, but this ``second level reputation'' 
is easier to manipulate and we will discard it for now.


The classical PoW is modulated by a parameter, the \textit{difficulty}, which is adjusted so that the 
validated blocks flow at a constant rate. The \textit{difficulty} is the same for all miners. If we are able 
to assign some type reputation to miners, it is natural to require for those that are reputed a lower 
difficulty. This is how we can implement a reputation bonus. This type of reasoning  is a standard procedure in
everyday dealings. For example, we trust more a dealer that we already know and that has a verifiable past track. 
The security requirements for engaging in trading with him is lower that with an unknown dealer
that may well be a scammer. 


The novel idea at the root of the reputation system presented is to have a variable  difficulty 
parameter for miners. In this way we can improve 
the protocol with a reputation ingredient: We assign a reputation bonus to miners 
that mined blocks in the past, that is, to anyone that has in the past contributed to the security of the 
network by providing energy and work. Because
of his verifiable past track, we allow this miner, to validate blocks with a lower difficulty, 
hence with a lower use of energy.


The nature of reputation makes that this reputation bonus should decay with time. 
Common sense indicates a higher trust in a 
dealer with recent successful deals that someone else with the same successful deals but from many years ago.
In the next section we propose a mathematical implementation of these natural ideas.

\section{Proof-of-Reputation.}\label{sec:PoR}

\subsection{Reputation bonus sequence.}

We consider a positive sequence $(\lambda_n)_{n\geq 0}$ with $0<\lambda_n <1$, $(\lambda_n)$ decreasing, and summable
with total sum $\Lambda$ with
\begin{equation}
0<\Lambda =\sum_{n\geq 0} \lambda_n <1
\end{equation}
The quantity $\Lambda$ is the total reputation bonus of the network and can be adjusted from $0$ to $1$ 
interpolating between a pure PoW and a pure PoR protocols. The quantity $\lambda_n$ is the reputation 
bonus earned by validating the block at depth $H-n$ where $H$ is the current blockchain height.

More precisely, for a miner $x$, we define
\begin{equation*}
\lambda_n(x) =\begin{cases} 
\lambda_n \ \ \text{if } x \text{ mined block } H-n \\ 0  \ \ \ \ \text{otherwise} 
\end{cases}
\end{equation*}

The total bonus reputation of miner $x$ is then
\begin{equation}
\Lambda(x)=\sum_{n\geq 0} \lambda_n(x) \ .
\end{equation}

Observe that the total reputation bonus of the network is 
\begin{equation}
\Lambda= \sum_{x} \Lambda(x) 
\end{equation}

We use the coinbase public payment address as the miner digital identity. In this way the bonus reputation of a miner 
can be computed from the data in the blockchain. If a miner wants to accumulate bonus reputation
he should use always the same payment address. We can easily add the possibility of multiple payment addresses
by requesting a proof of ownership of these other addresses by signing simultaneously with the other addresses
for which the miner wants to earn reputation bonus. This opens the possibility of selling reputation 
bonus that we may not want. Also a simpler way of having a reputation market would be to issue 
a reputation token with each payment in the coinbase transaction, but this is not in the spirit of the proposal presented in 
this article because we want to prevent anyone with large ressources to launch an immediate 51\% attack.

\subsection{Taylored difficulty.}

The output of the Proof-of-Work is in the range of SHA256 which is a 256 binary number. The \textit{target value}
$0\leq D \leq 2^{256}-1$ is the threshold under which the output is successful. The empiral pseudo-randomness 
property of SHA256, shows that the average number of iterations before finding a solution
is $d=2^{256}/D$ that we define as the difficulty. Thus, for a given hardware, the average energy needed for a validation  
is proportional to the difficulty $d=2^{256}/D$.

%
If miner $x$ has a reputation bonus $\Lambda(x)$ his new taylored difficulty will be
\begin{equation*}
d(x)=d(1-\Lambda(x)) 
\end{equation*}
and target value
\begin{equation*}
D(x)=\frac{D}{1-\Lambda(x)} 
\end{equation*}

Thus, with a reputation bonus $\Lambda(x)$ the miner saves a percentage $\Lambda(x)$ of energy to ensure the same
security. Assuming that all miners use their reputation bonus, we can achieve the same level of security saving
a percentage $\Lambda$ of the total energy. We can also argue that for the same energy consumption, we increase the 
security (measured in energy) by a factor $(1-\Lambda)^{-1}$.

\subsection{Structure of the reputation bonus sequence.}

As explained before, it is natural to request that the reputation bonus decays over time, 
this is why we want the sequence 
$(\lambda_n)$ to be decreasing. It is also natural to request decay uniform over time. This leads
to the natural choice of
\begin{equation*}
\lambda_n =\lambda_0  e^{-\chi n} 
\end{equation*}
where $\chi >0$ is the decay constant. We can request that the reputation bonus is halved every year. 
For example, in the case of Bitcoin protocol this leads to a value
$$
\chi= \frac{\log 2}{52560} \approx 1.31877. 10^{-5}
$$
which leads to a decay of $0.19\%$ per day of the bonus reputation.

Then we have 
\begin{equation*}
\Lambda = \frac{\lambda_0}{1-e^{-\chi}} 
\end{equation*}
If we set a target energy savings of 30\%, that is $\Lambda =0.3$, then for the previous example 
with Bitcoin we get
$$
\lambda_0 \approx 3.956.10^{-6}
$$

It also makes sense to implement a reputation bonus sequence $(\lambda_n)$ that does not decay 
exponentially initially, but only after some period of time (like a year for example). Several other 
reasonable choices are possible.

\section{Security analysis.}

\subsection{$51$\% attack.}

We consider a miner $x$ owning a fraction $0<q<1$ of the hashrate. In the long run he can expect 
to earn a reputation bonus of $\Lambda'=q\Lambda$ since he mines on average a proportion $q$ of the blocks.
If everybody else uses his reputation bonus, the apparent hashrate will still be $q$. Then the $51$\%
can still only be performed when $q>1/2$.
But, in the most unfavorable situation where no one else uses 
its reputation bonus, his apparent hashrate is higher. We have

\begin{proposition}
Assuming the previous conditions, the apparent hashrate of the only miner using his reputation bonus 
is $\frac{h'}{1-\Lambda'}=\frac{h'}{1-q\Lambda}$ where 
$h'$ is his hashrate.
\end{proposition}
\begin{proof}
Let $\mathbf T'$, resp. $\mathbf T$, be the random variable of the time that 
it takes to miner using his reputation bonus, resp. the rest of the miners, to mine a block. Then
we have
\begin{align*}
\mathbb E [\mathbf T'] &= \frac{d(1-\Lambda')}{h'} =\frac{d}{h'/(1-q\Lambda)} \\
\mathbb E [\mathbf T] &= \frac{d}{h}
\end{align*}
where $h$ is the hashrate of the rest of the miners in the network.
\end{proof}

Now we prove:

\begin{proposition}
Assuming the previous conditions, the miner only miner using his reputation bonus succeeds in a $51$\% attack when
$$
q > q_0=\frac{1}{2}- \frac{\sqrt{4+\Lambda^2}-2}{2\Lambda}
$$
\end{proposition}

\begin{proof}
The condition to succeed a  51\% attack exactly means that it is more likely 
that the attacker miner mines a block earlier than the rest of the network, i.e.
$$
\mathbb P[\mathbf T'<\mathbf T] >1/2
$$
The random variables $\mathbf T$ and $\mathbf T'$ follow exponential distributions 
(see for example \cite{GPM}) with parameters $\alpha'$, resp $\alpha$,
\begin{align*}
\alpha' &=\frac{1}{\mathbb E[\mathbf T']}=\frac{h'}{d(1-q\Lambda)} \\
\alpha  &= \frac{1}{\mathbb E[\mathbf T]}=\frac{h}{d}
\end{align*}
We have
\begin{equation*}
\mathbb P[\mathbf T'<\mathbf T] = \frac{\alpha'}{\alpha'+\alpha} =\frac{\frac{h'}{d(1-q\Lambda)}}{\frac{h}{d}+\frac{h'}{d(1-q\Lambda)}}=
\frac{h'}{(1-q\Lambda)h+h'}=\frac{\frac{h'}{h+h'}}{1-q\Lambda \frac{h}{h+h'}} =\frac{q}{1-q\Lambda (1-q)}
\end{equation*}
Therefore, the condition $\mathbb P[\mathbf T'<\mathbf T] >1/2$ becomes
$$
q>\frac{1}{2} (1-\Lambda q(1-q))
$$
which is equivalent to
$$
\Lambda q^2 -(\Lambda +2)q +1 <0
$$
The discriminant is $\Delta =1+\Lambda^2$ and the quadratic polynomial has one root $>1$ and 
another root $0<q_0<1$ given by the formula in the statement.
\end{proof}

For a value of $\Lambda = 0.3$ this means that $q>q_0=0.4627\ldots$. This is not significantly lower than 50 \%.
The attacker must prepare the attack with time ahead in order to gain the 4\% reputation bonus. We clearly achieve 
the same or betterlevel of security.
From the view pont of the security of 
the network, in some sense, we trade pure hashrate for time and work previously performed.

\subsection{Reputation attack.}

In the Proof-of-Reputation protocol, we will have some miners 
that have acquired along their PoW history
some reputation. The new miners are free to gain this reputation bonus through their work. As we know, most miners 
pool together in ``pools'' in order to mitigate the volatility of their income because a block validation
is a rare event for a low hashrate. Thinking about the pool system in place, we can see a dangerous scenario in the PoR protocol: 
The pool with higher reputation will be more profitable and will attract all the miners. Therefore
will quickly reach more than 50\% hashrate. This can compromise the security of the network.

As we can see, what is really happening is that miners joining the pool take advantage of 
the reputation bonus previously earned by the pool in order to lower their mining difficulty. We need another key idea for 
a proper functioning of a reputation system: We should be able to certify that whoever benefits from any 
reputation bonus is the same virtual identity that has done the past work that provided the reputation 
bonus. In the present situation, there is a cryptographic solution to this problem: \textit{Certified Mining}.

\subsection{Certified mining.}

 We describe \textit{Certified Mining} for a PoW algorithm based on a 
hash function as the one for Bitcoin. We remind that Bitcoin's Pow consists in formatting a block, 
and hash (more precisely, double SHA256) its 
header which contains critical information about the block and the transactions, in particular 
it contains the Merkle
root of the transactions in the body of the block. A variation of the nounce, the extra nounce, or the 
format of the body of the block, provides new output hashes and the goal is to get one hash below the difficulty 
threshold. In their regular operations, a mining pool formats the blocks and provide 
the headers to hash to the miners.
They do format the block with a payment address of the coinbase transaction that they do control (this is where
the block reward is paid). The solution to avoid the reputation usurpation is simple: 
Once the block is formatted as it is done
classically, we require the header of the block to be signed with the secret key of the payment coinbase 
address. This signature is appended at the end to the block header. Each time the nonce is changed, 
a new signature is required. The bulk of the PoW is then the signatures more than the hashes, and this work 
can only be provided by the entity in control of the payment address.
In some sense, the payment public address of the 
coinbase transaction is the id of the miner (he can also have several addresses that can be used to cumulate 
reputation bonuses and provide a proof of ownership).

Obviously, for security reasons, the pool cannot share the secret key of the payment address, and in general, 
only the owner of the payment address can perform the signature. Observe that the signature includes the 
nonce, extra nonce, etc and the main computational task is to sign the header (the double hash is much faster
to do). Hence this modification of the PoW does not allow to subcontract the PoW taking advantage of a 
reputation bonus earned by someone else.


This procedure seems to make impossible a pooling system, but this is not so. A pooling system 
is viable where each one contributes and gets a share weighted by its verifiable reputation. 

\subsection{Transition to a PoR system.}

It is natural to imagine how will be a transition to PoR. Obviously this requires a hardfork of the protocol that 
changes the mining algorithm to allow a taylored difficulty for different miners.
The authors do not advocate for any hardfork of the Bitcoin 
protocol whose ideal state is converging to ossification as a payment protocol. If such a hardfork is implemented, 
it needs to be performed with extraordinary caution. 
The transition to PoR can indeed be realized gradually. 
For example, the value of $\Lambda$ could be increased gradually during a large period of time (for example, 10 years) 
from $0$ to its limit value ($0.3$ in our numerical example). 
For these conservative values, there is little risk for Bitcoin security.

\end{document}